\documentclass[onecolumn,noversion,nonote]{cdcarticle}

		\def\MODE{1}
		\def\CDC{1}

\usepackage{amsmath,amsthm,amssymb}
\usepackage{tikz}
\usepackage{graphicx}
\usepackage{enumerate}
\usepackage[hidelinks]{hyperref}
\usepackage{caption}

\usetikzlibrary{matrix,fit,backgrounds,plotmarks,calc}
\usepackage{pgfplots}

\if\MODE1
	\pgfplotsset{width=10cm,height=7.5cm,compat=newest}
\else
	\pgfplotsset{width=6.5cm,height=5.5cm,compat=newest}
\fi

\tikzset{dashdot/.style={dash pattern=on 1pt off 2.5pt on 4.5pt off 2.5pt}}
\pgfplotscreateplotcyclelist{gray marks}{%
every mark/.append style={fill=gray,scale=0.8},mark=*\\%
every mark/.append style={fill=gray},mark=diamond*\\%
every mark/.append style={fill=gray},mark=triangle*\\%
every mark/.append style={fill=gray,scale=0.8},mark=square*\\%
}
\pgfplotscreateplotcyclelist{color dashes}{%
blue,very thick\\
red,very thick\\
green!40!black,very thick\\
red,very thick,dashdot\\
blue,very thick,dashed\\
red,very thick,dashed\\
}

\pgfplotsset{
  tick label style = {font=\footnotesize},
  every axis label = {font=\footnotesize},
  legend style = {font=\footnotesize},
  label style = {font=\normalsize}
}

\newtheorem{thm}{Theorem}

\newtheorem{assumption}[thm]{Assumption}

\newtheorem{prop}[thm]{Proposition}

\renewenvironment{proof}{{\noindent\bf Proof.}}{ \hfill ~\qed}
\def\qed{\rule[0pt]{5pt}{5pt}\par\medskip}

\newcommand{\tp}{\mathsf{T}}								
\newcommand{\bmat}[1]{\begin{bmatrix}#1\end{bmatrix}}	
\newcommand{\bmatshort}[1]{\left[ \begin{smallmatrix}#1\end{smallmatrix} \right]}		

\DeclareMathOperator{\diag}{\mathrm{diag}}
\DeclareMathOperator*{\argmin}{\arg\min}

\begin{document}
\if\CDC1\else\begin{frontmatter}\fi
\if\CDC2
	\title{Compositional Performance Certification of Interconnected\\
			 Systems using ADMM\tnoteref{t1}}
	
	\author[mech]{Chris~Meissen\corref{cor1}}
	\ead{cmeissen@berkeley.edu}
	\author[mech]{Laurent~Lessard}
	\ead{lessard@berkeley.edu}
	\author[eecs]{Murat~Arcak}
	\ead{arcak@berkeley.edu}
	\author[mech]{Andrew~Packard}
	\ead{apackard@berkeley.edu}
	
	\tnotetext[t1]{This work was supported in part by NSF under grant EECS-1405413 and NASA under Grant No.~NRA~NNX12AM55A entitled ``Analytical Validation Tools for Safety Critical Systems Under Loss-of-Control Conditions'', Dr. Christine Belcastro technical monitor. Any opinions, findings, and conclusions or recommendations expressed in this material are those of the author and do not necessarily reflect the views of the NSF or NASA.}
	
	\cortext[cor1]{Corresponding author}
	\address[mech]{Department of Mechanical Engineering, University of California, Berkeley, CA~94720, USA.}
	\address[eecs]{Department of Electrical Engineering, University of California, Berkeley, CA~94720, USA.}
\else
	\title{Compositional Performance Certification of \\Interconnected
			Systems using ADMM}
	\author{Chris Meissen \quad Laurent Lessard \quad Murat Arcak \quad Andrew Packard}
	\note{Submitted to Automatica}
	\maketitle
\fi

\begin{abstract}
A compositional performance certification method is presented for interconnected systems using subsystem dissipativity properties and the interconnection structure. A large-scale optimization problem is formulated to search for the most relevant dissipativity properties. The alternating direction method of multipliers (ADMM) is employed to decompose and solve this problem, and is demonstrated on several examples.
\end{abstract}

\if\CDC2
\begin{keyword}
compositional analysis \sep interconnected systems \sep ADMM \sep dissipative dynamical systems
\end{keyword}
\end{frontmatter}
\fi

\section{Introduction}

In this paper, compositional analysis is used to certify performance of an interconnection of subsystems as depicted in Figure~\ref{fig:IntSysIO}. The $G_i$ blocks are known subsystems mapping $u_i \mapsto y_i$ and $M$ is a static matrix that characterizes the interconnection topology. The goal of compositional analysis is to establish properties of the interconnected system using only properties of the subsystems and their interconnection. Henceforth, the term ``local'' is used to refer to properties or analysis of individual subsystems in isolation. Likewise, ``global'' refers to the entire interconnected system.

\begin{figure}[ht]
\centering
\begin{tikzpicture}[thick,auto,>=latex,node distance=1.6cm]
\tikzstyle{block}=[draw,rectangle,minimum width=2.5em, minimum height=2.5em]
\def\x{1.0}		
\def\y{0.2}		
\def\z{1.4}		
\node [block](G){$M$};
\node [block,above of=G](D){$\addtolength{\arraycolsep}{-0.2em}\begin{matrix}
\,{G}_1 & &\\[-0.15em]&\ddots&\\[-0.15em]& & {G}_N \end{matrix}$};
\draw [<-] (G.east)+(0,\y) -- +(\x,\y) |- node[swap,pos=0.25]{$y$} (D);
\draw [<-] (G.east)+(0,-\y) -- +(\z,-\y) node [anchor=west]{$d$};
\draw [->] (G.west)+(0,\y) -- +(-\x,\y) |- node[pos=0.25]{$u$} (D);
\draw [->] (G.west)+(0,-\y) -- +(-\z,-\y) node [anchor=east]{$e$};
\end{tikzpicture}
\caption{Interconnected system with input $d$ and
 output $e$.\label{fig:IntSysIO}}
\end{figure}
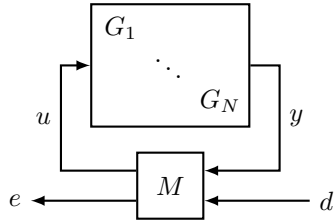

Local behavior and global performance are cast and quantified in the framework of dissipative systems~\cite{willems72}; specifically the case with quadratic supply rates. The global supply rate is specified by the analyst and dictates the system performance that is to be verified. For example, supply rates can be chosen to characterize $L_2$-gain, passivity, output-strict passivity, etc., for the input-output pair ($d$,$e$). A storage function is then sought to certify dissipativity with respect to the desired supply rate. See Section~\ref{sec:prelim} for definitions of storage functions and supply rates.

A conventional approach to compositional analysis, as presented for example  in~\cite{anderson11, dash07, sandell78, vidyasagar81, willems72}, is to establish individual supply rates (and storage functions) for which each subsystem is dissipative. Then, a storage function certifying dissipativity of the interconnected system is sought as a combination of the subsystem storage functions.

The method presented here is less conservative because the local supply rates (and storage functions) are optimized with regards to their particular suitability in certifying global properties. Thus, the local certificates are automatically generated, as opposed to being preselected.

Optimizing over the local supply rates (and storage functions) to certify stability of an interconnected system was first introduced in~\cite{ufuk09}, with the individual supply rates constrained to be diagonally-scaled induced $L_2$-norms. This perspective, coupled with dual decomposition, gave rise to a distributed optimization algorithm. We generalize this approach in several ways: certifying dissipativity (rather than stability) of the interconnected system with respect to a quadratic supply rate; searching over arbitrary quadratic supply rates for the local subsystems; and employing ADMM~\cite{boyd11} to decompose and solve the resulting problem. 

The ADMM algorithm exposes the distributed certification as a convergent negotiation between parallelizable, local  problems for each subsystem, and a global problem. Each local problem receives a proposed supply rate from the global problem and solves an optimization problem certifying dissipativity of the corresponding subsystem with a supply rate close to the proposed one. The global problem, with knowledge of the interconnection $M$ and the updated supply rates, solves an optimization problem to certify dissipativity of the interconnected system and proposes new supply rates.

 In \cite{acc14} the method presented here was applied to linear systems and ADMM was compared to other distributed optimization methods. In \cite{cdc14} this method was extended to nonlinear systems using sum-of-squares (SOS) optimization. Additionally, \cite{cdc14} generalized this approach to systems that are \emph{equilibrium-independent dissipative}~\cite{hines11}.

This paper unifies and expands on the conference papers \cite{cdc14,acc14}. A new theorem shows the proposed method is equivalent to searching for an additively separable storage function for interconnections of linear subsystems. We also demonstrate that the proposed method is tractable and more efficient for large systems than conventional techniques. An extension of the proposed method using integral quadratic constraints is included to allow frequency dependent properties of the subsystems. New examples are presented to demonstrate the results. The convergence properties of ADMM are described and shown to hold for this application.

\section{Preliminaries}
\label{sec:prelim}

\paragraph{Dissipative dynamical systems \cite{willems72}} Consider a time-invariant dynamical system:
\begin{equation} \label{eq:Sys}
\begin{aligned}
\dot{x}(t) &= f(x(t),u(t)),		& f(0,0) &= 0 \\
y(t) &= h(x(t),u(t)),				& h(0,0) &= 0
\end{aligned}
\end{equation}
with $x(t) \in \mathbb{R}^n$, $u(t) \in \mathbb{R}^m$, $y(t) \in \mathbb{R}^p$. 
A \emph{supply rate} is a function $w: \mathbb{R}^m \times \mathbb{R}^p \to \mathbb{R}$. A system of the form~\eqref{eq:Sys} is \emph{dissipative} with respect to a supply rate~$w$ if there exists a differentiable, nonnegative function $V: \mathbb{R}^n \to \mathbb{R}_+$ such that $V(0) = 0$ and
 \begin{align}\label{eq:DIE}
\nabla V(x)^\tp f(x, u) - w(u,h(x,u)) \leq 0
\end{align}
for all $x \in \mathbb{R}^n$ and $u \in \mathbb{R}^m$. Equation~\eqref{eq:DIE} is referred to as the \emph{dissipation inequality} and $V$ as a \emph{storage function}. 

\paragraph{Equilibrium-independent dissipative (EID) systems \cite{burger13, hines11}}
Consider a system of the form
 \begin{equation} \label{eq:SysEID}
 \begin{aligned}
 \dot{x}(t) = f(x(t),u(t)) \qquad
 y(t) = h(x(t),u(t))
 \end{aligned}
 \end{equation}
where there exists a nonempty set $\mathcal{X}^\star \subseteq \mathbb{R}^n$ such that for each $x^\star \in \mathcal{X}^\star$ there exists a unique $u^\star \in \mathbb{R}^m$ such that $f(x^\star,u^\star) = 0$. The \emph{equilibrium state-input map} is then defined as 
 \begin{equation*}
 k_u(x):\mathbb{R}^n \rightarrow \mathbb{R}^m \text{ such that } u^\star = k_u(x^\star).
 \end{equation*}
 The system~\eqref{eq:SysEID} is \emph{EID} with respect to a supply rate $w$ if there exists a nonnegative storage function $V: \mathbb{R}^{2n} \rightarrow \mathbb{R}_+$ such that $V(x^\star,x^\star) = 0$ and 
\begin{align} \label{eq:DIEeid}
\begin{split}
\nabla_{x}V(x,x^\star)^\tp f(x,u)-w(u-u^\star,y-y^\star) \leq 0
\end{split}
\end{align}
for all $x^\star \in \mathcal{X}^\star$, $x \in \mathbb{R}^n$, and $u \in \mathbb{R}^m$ where $u^\star=k_u(x^\star)$, $y = h(x,u)$, and $y^\star = h(x^\star, u^\star)$.

This definition ensures dissipativity with respect to any possible equilibrium point rather than a particular point. This is advantageous for compositional analysis, since the equilibrium of an interconnection may be hard to compute.

\paragraph{Integral Quadratic Constraints (IQCs)~\cite{megretski97}} 
IQCs are a generalization of the dissipativity framework that capture frequency dependent properties of a system. Let $(\hat A,\hat B,\hat C,\hat D)$ be the realization of a stable LTI system~$\Psi$ with state $\eta$ and $X$ be a real symmetric matrix.  Then \eqref{eq:Sys} satisfies the IQC defined by~$\Pi = \Psi^* X \Psi$ if there exists a nonnegative storage function $V(x, \eta)$ such that $V(0,0) = 0$ and 
\begin{align}\label{eq:DIEiqc}
\nabla_x V&(x, \eta)^\tp  f(x,u) + \nabla_{\eta} V(x, \eta)^\tp 
			 \left( \hat A \eta + \hat B \bmat{ u \\ y } \right) \notag \\
			&\leq 
			 \left( \hat C \eta + \hat D \bmat{ u \\ y } \right)^\tp X 
			 \left( \hat C \eta + \hat D \bmat{ u \\ y } \right)
\end{align}
for all $x \in \mathbb{R}^n$ and $u\in \mathbb{R}^m$ where $y = h(x,u)$. In fact, more is true, \eqref{eq:DIEiqc} implies that for all $u \in  L_{2e}$, the space of signals that are square integrable on all finite intervals, the signal $z := \Psi \bmatshort{u \\y}$ satisfies $\int_0^T z^\tp X z dt \geq 0$ for all $T>0$ with $x(0) = 0$ and $\eta(0) = 0$. Dissipativity is recovered when $\Psi = I_{m+p}$.

\paragraph{SOS programming \cite{parillo00}}
For polynomial systems certifying dissipativity can be relaxed to a semidefinite program (SDP) searching for storage functions that are SOS polynomials. 

 Suppose that $f$ and $h$ in \eqref{eq:Sys} are polynomials. Let $\mathbb{R}[x]$ ($\Sigma[x]$) be the set of polynomials (SOS polynomials) in $x$. Then certification of dissipativity with respect to a polynomial supply rate, $w$, can be relaxed to the SOS feasibility program:
\begin{align}
\label{eq:SOSDiss}
\begin{aligned}
V(x) &\in \Sigma[x] \\
-\nabla V(x)^\tp f(x,u) + w(u,y) &\in \Sigma[x,u].
\end{aligned}
\end{align}
Similarly, as presented in \cite{hines11}, certifying polynomial systems are EID can be relaxed to:
\begin{align}\label{eq:SOSDissEID}
\begin{aligned}
V(x,x^\star) &\in \Sigma[x, x^\star] \\
r(x, u, x^\star, u^\star) &\in \mathbb{R}[x,u,x^\star,u^\star] \\
-\nabla_x V(x,x^\star)^\tp f(x,u) + w(u-u^\star,y-y^\star) \hspace{-20mm}& \\
+\,r(x, u, x^\star, u^\star)f(x^\star, u^\star) &\in \Sigma[x,u, x^\star, u^\star].
\end{aligned}
\end{align}
If each state has rational polynomial dynamics,
\begin{equation*}
\dot{x}_i = f_i(x,u) = \frac{p_i(x,u)}{q_i(x,u)} \quad \text{ for } i = 1, \dots, n
\end{equation*}
where $p_i \in \mathbb{R}[x,u]$ and  $q_i - \epsilon \in \Sigma[x,u]$ for  $\epsilon > 0$, then certifying dissipativity of the system with respect to a polynomial supply rate, $w$, can be relaxed to:
\if\MODE1
\begin{align}
\label{eq:ratSOSDiss}
\begin{aligned}
V(x) &\in \Sigma[x] \\[-5pt]
-\sum_{i=1}^n\nabla_{x_i} V(x) p_i(x,u) \prod_{j \neq i} q_j(x,u)+\prod_{i=1}^n q_i(x,u) w(u,y) &\in \Sigma[x, u].
\end{aligned}
\end{align}
\else
\begin{align}
\label{eq:ratSOSDiss}
\begin{aligned}
V(x) &\in \Sigma[x] \\[-5pt]
-\sum_{i=1}^n\nabla_{x_i} V(x) p_i(x,u) \prod_{j \neq i} q_j(x,u)\hspace{3mm}& \\[-5pt]
+\prod_{i=1}^n q_i(x,u) w(u,y) &\in \Sigma[x, u].
\end{aligned}
\end{align}
\fi
Similarly to the polynomial case, certifying rational polynomial systems are EID can also be formulated as an SOS feasibility program. Furthermore, certifying a polynomial or rational polynomial system satisfies an IQC can be formulated as a SOS feasibility program.

 Consider the interconnected system in Figure~\ref{fig:IntSysIO} where the subsystems $G_1,\dots,G_N$ are known and have dynamics of the form~\eqref{eq:Sys}. Each subsystem $G_i$ is characterized by $(f_i,h_i)$ with $x_i(t)\in\mathbb{R}^{n_i}$, $u_i(t)\in\mathbb{R}^{m_i}$, and $y_i(t)\in\mathbb{R}^{p_i}$. Define $n:= n_1+\dots+n_N$.
 
 The static interconnection $M \in \mathbb{R}^{m \times p}$ relates
\begin{equation}\label{eq:M_interconnection}
\bmat{ u \\ e } = M\bmat{y \\ d } 
\end{equation}
where $d(t) \in \mathbb{R}^{p_{d}}$, $e(t) \in \mathbb{R}^{m_{e}}$, $m = m_1 +\dots+ m_{N}+m_e$, and  $p = p_1 +\dots+ p_{N}+p_d$. We assume the interconnection is well-posed: for any $d \in L_{2e}$ and initial condition $x(0)\in\mathbb{R}^n$ there exist unique $e, u, y \in L_{2e}$ that causally depend on $d$.

The global and local supply rates are assumed to be quadratic forms. In particular, the global supply rate, which is specified by the analyst, is
\begin{equation}\label{eq:quadratic_supply_rates}
\bmat{d \\ e}^\tp W \bmat{d \\ e} 
\end{equation}
where $W$ is a real symmetric matrix. The local supply rates are 
\begin{equation}
\bmat{ u_i \\ y_i }^\tp X_i \bmat{ u_i \\ y_i } 
\end{equation}
where $X_i$ are real symmetric matrices. 

Certifying dissipativity of the interconnected system is a feasibility problem. In order to state this problem we first define the \emph{local} constraint sets as
\if\MODE1
	\begin{align}\label{def:Li}
	\mathcal{L}_i &:= \Biggl\{X_i \,\Bigg|\, \text{ the $i$-th subsystem is dissipative w.r.t. } \bmat{ u_i \\  y_i }^\tp X_i \bmat{ u_i \\ y_i } \Biggr\}
	\end{align}
\else
	\begin{align}\label{def:Li}	
	\mathcal{L}_i &:= \Biggl\{X_i \,\Bigg|\, \text{ the $i$-th subsystem is dissipative} \notag\\ & \text{w.r.t. the supply rate } \bmat{ u_i \\  y_i }^\tp X_i \bmat{ u_i \\ y_i } \Biggr\}
	\end{align}
\fi
and the \emph{global} constraint set as
\begin{align} \label{def:G}
	&\mathcal{G} := \Biggl\{ X_1, \dots, X_N  \,\Bigg|\,
	 \bmat{M \\ I_p}^\tp\!\!\! P_\pi^\tp Q P_\pi \bmat{M \\ I_p}\preceq 0 \Biggr\}
\end{align}
where $Q = \diag(X_1, \dots, X_N, -W)$ and $P_\pi$ is a permutation matrix defined by
\begin{align*}
\bmat{ u_1 \\ y_1 \\ \vdots \\ u_N \\ y_N \\ d \\ e} = P_\pi \bmat{u \\ e \\ y \\ d}. 
\end{align*}
 
The following elementary proposition gives conditions for certifying dissipativity of the interconnected system. This proposition is a direct extension of previous results given in~\cite{moylan78, vidyasagar81}.

\begin{prop}\label{prop:generalized_inequality}
Consider $N$ subsystems interconnected according to~\eqref{eq:M_interconnection}, with a global supply rate of the form~\eqref{eq:quadratic_supply_rates}. If there exist $X_1,\dots,X_N$ satisfying
\begin{align}
\label{eq:primal}
\begin{aligned}
 & X_i \in \mathcal{L}_i\quad\text{for }i=1,\dots,N \\
& (X_1,\dots,X_N) \in \mathcal{G}
\end{aligned}
\end{align}
then the interconnected system is dissipative with respect to the global supply rate. A storage function certifying global dissipativity is $V(x_1,\dots,x_N) := \sum_{i=1}^N V_i(x_i)$ where each $V_i(x_i)$ is a storage function certifying local dissipativity according to~\eqref{def:Li}.
\end{prop}

\begin{proof} Multiplying the LMI in~\eqref{def:G} on the left by $\bmatshort{y \\ d}^\tp$ and on the right by $\bmatshort{y \\ d}$ and making use of~\eqref{eq:M_interconnection}, we obtain
\begin{align} \label{eq:DissCert}
 \sum_{i=1}^N \bmat{ u_i \\ y_i }^\tp X_i
 \bmat{ u_i \\ y_i } - \bmat{ d \\ e }^\tp W \bmat{ d \\ e } \leq 0.
 \end{align} 
Since $X_i \in \mathcal{L}_i$ a storage function $V_i$ exists such that
\begin{align} \label{eq:localDIE}
\nabla V_i(x_i)^\tp f_i(x_i, u_i) - \bmat{ u_i \\ y_i }^\tp X_i \bmat{ u_i \\ y_i } \leq 0 \hspace{0.1in} 
\end{align}
for all $x_i\in \mathbb{R}^n$, $u_i \in \mathbb{R}^m$, and  $y_i = h_i(x_i,u_i)$.
Adding to~\eqref{eq:DissCert} the local dissipativity inequalities~\eqref{eq:localDIE} for each subsystem, we obtain
\begin{align} \label{eq:diss}
 \sum_{i=1}^N \nabla V_i(x_i)^\tp f_i(x_i,u_i) - \bmat{ d \\ e }^\tp W \bmat{ d \\ e } \leq 0
\end{align}
which certifies dissipativity with respect to the global supply rate.
\end{proof}

Proposition~\ref{prop:generalized_inequality} certified global dissipativity from subsystem dissipativity. The following theorem states that this compositional approach is not overly conservative: for linear systems the existence of a separable quadratic storage function for the interconnection is \emph{equivalent} to the existence of supply rates satisfying~\eqref{eq:primal}. Related results in~\cite{colgate88,kerber11} show that, under certain assumptions on the interconnection, passivity of an interconnected system is equivalent to passivity of the subsystems.  First, we make a mild assumption on the interconnection.

\begin{assumption}
\label{assum:Mind} The block diagonal elements of $M$ mapping $y_i \rightarrow u_i$ are zero and the rows of $M$ mapping $\bmatshort{y \\d} \rightarrow u_i$ are linearly independent for each $i$.
\end{assumption}
The first part of Assumption~\ref{assum:Mind} implies the subsystems don't have self-feedback loops. The second part implies that no elements of the input~$u_i$ for each subsystem are identical for all $\bmatshort{y \\ d} \in \mathbb{R}^{m}$.

\begin{thm} \label{thm:equivLin}
Consider $N$ linear subsystems
\begin{subequations}\label{subsys}
\begin{align}
\dot{x}_i &= A_ix_i + B_iu_i \\ \label{subsysb}
y_i &= C_ix_i
\end{align}
\end{subequations}
interconnected according to~\eqref{eq:M_interconnection} where $M$ satisfies Assumption~\ref{assum:Mind}. Suppose a global supply rate of the form~\eqref{eq:quadratic_supply_rates} is given. The following are equivalent:
\begin{enumerate}[(i)]
\item \label{cond1} There exists  a separable quadratic storage function of the form $V(x_1,\dots, x_N) = \sum_{i=1}^N x_i^\tp P_i x_i$ certifying dissipativity of the interconnected system.
\item \label{cond2} Each subsystem is dissipative and the associated supply rate matrices $X_1,\dots, X_N$ satisfy the global constraint $\mathcal{G}$. Dissipativity of the subsystems can be certified with storage functions of the form $V_i(x_i) = x_i^\tp P_i x_i$.
\end{enumerate} 
\end{thm}
\begin{proof} 
See appendix.
\end{proof}

We next extend the results of Proposition~\ref{prop:generalized_inequality} to systems that are EID or satisfy IQCs.

\paragraph{Extension to EID systems} Assume there exists a nonempty set $\mathcal{X}^\star \subseteq \mathbb{R}^n$ such that for each $x^\star \in \mathcal{X}^\star$ there is a unique $d^\star$ such that $f_i(x_i^\star, u_i^\star) = 0$ where $y_i^\star = h_i(x_i^\star, u_i^\star)$ and $\bmat{u^\star \\ e^\star} = M \bmat{y^\star \\ d^\star}$.

The global and local supply rates must be modified to depend on $u^\star$, $y^\star$, $d^\star$, and $e^\star$. Specifically, the global supply rate becomes
\begin{equation}\label{eq:quadratic_supply_ratesEID}
\bmat{d-d^\star \\ e-e^\star}^\tp W \bmat{d-d^\star \\ e-e^\star}
\end{equation}
and the local supply rates are
\begin{equation}
\bmat{ u_i-u_i^\star \\ y_i-y_i^\star }^\tp X_i \bmat{ u_i-u_i^\star \\ y_i-y_i^\star }
\end{equation}
where $W$ and $X_i$ are real symmetric matrices. 

For each subsystem we must determine a supply rate $X_i$ such that the subsystem is EID. Therefore, the local constraints sets are
\if\MODE1
	\begin{align}\label{def:LiEID}
	\mathcal{L}_i &:= \Bigl\{X_i \,\Big|\, \text{ the $i$-th subsystem is EID w.r.t. } \bmat{ u_i-u_i^\star \\  y_i-y_i^\star }^\tp X_i \bmat{ u_i-u_i^\star \\ y_i-y_i^\star } \Biggr\}.
	\end{align}
\else
	\begin{align}\label{def:LiEID}	
	\mathcal{L}_i &:= \Biggl\{X_i \,\Bigg|\, \text{ the $i$-th subsystem is EID w.r.t. the} \notag\\ & \text{supply rate } \bmat{ u_i-u_i^\star \\  y_i-y_i^\star }^\tp X_i \bmat{ u_i-u_i^\star \\ y_i-y_i^\star } \Biggr\}.
	\end{align}
\fi
The global constraint set $\mathcal{G}$ is unchanged.

Proposition~\ref{prop:generalized_inequality} is directly applicable to this formulation, thus a feasible solution to~\eqref{eq:primal} with the local constraint sets defined as in~\eqref{def:LiEID} implies that the interconnected system is EID with respect to~\eqref{eq:quadratic_supply_ratesEID}.

\paragraph{Extension to IQCs}
 We extend this approach to IQCs by redefining the local constraint sets as
\begin{align}\label{eq:LiIQC}
	\mathcal{L}_i := &\Big\{X_i \, \Big|\, \exists V_i(x_i,\eta_i) \succeq 0 \text{ s.t. \eqref{eq:DIEiqc} holds} \Big\}
\end{align}
where $\eta_i$ is the state of the stable linear system $\Psi_i$ specified by the analyst and the symmetric matrix $X_i$ is a decision variable. Let $(\hat{A}_i, \hat{B}_i, \hat{C}_i, \hat{D}_i)$ be a state space realization of $\Psi_i$.

Rather than certify performance with respect to a global supply rate, we generalize this to certifying performance with respect to a global IQC of the form $\Pi_\textsc{w} := \Psi_\textsc{w}^* W \Psi_\textsc{w}$ where $\Psi_\textsc{w}$ is a stable linear system with realization $(\hat{A}_\textsc{w}, \hat{B}_\textsc{w}, \hat{C}_\textsc{w}, \hat{D}_\textsc{w})$ and $W$ is a real symmetric matrix, both specified by the analyst. 

Let $\bmatshort{y \\ d}$ be the input to a stable linear system $\Psi$ with the state space realization:
\begin{align*}
\hat{A} &:= \diag(\hat{A}_1, \dots, \hat{A}_N, \hat{A}_\textsc{w}) \\[-1mm]
\hat{B} &:= \diag(\hat{B}_1, \dots, \hat{B}_N, \hat{B}_\textsc{w})P_\pi \bmat{M\\I} \\[-1mm]
\hat{C} &:= \diag(\hat{C}_1, \dots, \hat{C}_N, \hat{C}_\textsc{w}) \\[-1mm]
\hat{D} &:= \diag(\hat{D}_1, \dots, \hat{D}_N, \hat{D}_\textsc{w})P_\pi \bmat{M\\I}.
\end{align*}
Then the global constraint set is defined as 
\begin{align} \label{def:GIQC}
\mathcal{G} := & \Biggl\{ X_1, \dots, X_N \,\Bigg|\,
	\exists \, P \succeq 0 \text{ such that } \\
&\bmat{\hat{A}^\tp P + P \hat{A} & P \hat{B} \\ \hat{B}^\tp P & 0} + \bmat{\hat{C}^\tp \\ \hat{D}^\tp} Q \bmat{\hat{C}^\tp \\ \hat{D}^\tp}^\tp \preceq 0 \Biggr\} \notag
\end{align}
 where $Q = \diag(X_1, \dots, X_N, -W)$, as in~\eqref{def:G}. The local and global constraints simplify to~\eqref{def:Li} and~\eqref{def:G}, respectively, if the IQCs are static ($\hat{A}$, $\hat{B}$, and $\hat{C}$ are of dimension zero) and $\hat{D} =P_\pi \bmatshort{M \\I}$.

\begin{prop}\label{prop:iqc}
Suppose that $X_1,\dots,X_N$ satisfy~\eqref{eq:primal} with the local and global constraints sets defined as \eqref{eq:LiIQC}-\eqref{def:GIQC}. Then the interconnected system satisfies the global IQC $\Pi_\textsc{w}$.
\end{prop}

\begin{proof} Let $\eta$ be the state and $z$ the output of $\Psi$. Defining $V(\eta) := \eta^\tp P\eta$ gives
\begin{align*}
\dot{V} (\eta,y,d) =   \bmat{\eta \\ y \\ d}^\tp \bmat{\hat{A}^\tp P + P\hat{A} & P\hat{B} \\ \hat{B}^\tp P & 0} \bmat{\eta \\ y \\ d}.
\end{align*}
The LMI in \eqref{def:GIQC} is equivalent to
\begin{align} \label{eq:iqcDIEinZ}
\dot{V} (\eta,y,d) + z^\tp Q z \leq 0
\end{align}
 for all $\eta$, $y$, and $d$ with $z = \hat{C}\eta+\hat{D}\bmatshort{y\\d}$. We partition $z$ such that $z_i$ is the output of $\Psi_i$ and $z_\textsc{w}$ is the output of $\Psi_\textsc{w}$. Then \eqref{eq:iqcDIEinZ} becomes
\begin{align} \label{eq:iqcDIEglob}
\dot{V} (\eta,y,d) + \sum_{i=1}^N z_i^\tp X_i z_i \leq z_\textsc{w}^\tp W z_\textsc{w}.
\end{align} 

Since each subsystem satisfies an IQC $\Pi_i$, we have $\dot{V}(x_i, \eta_i, u_i, y_i) \leq z_i^\tp X_i z_i$. Adding this to \eqref{eq:iqcDIEglob} gives
\begin{align*}
\dot{V} (\eta,y,d) + \sum_{i=1}^N \dot{V}_i(x_i, \eta_i, u_i, y_i)  \leq z_\textsc{w}^\tp W z_\textsc{w} 
\end{align*} 
certifying the interconnected system satisfies the global IQC $\Pi_\textsc{w}$.
\end{proof}
For computational considerations it is useful to understand the size of the LMI in \eqref{def:GIQC}. Suppose each subsystem has $n_y$ outputs and $n_u$ inputs and 
\begin{align}
\Psi_i(s) = \bmat{\psi^0_i(s)  \\ \vdots \\ \psi^{n_b}_i(s)}\otimes I_{n_y+n_u}
\end{align}
where $\otimes$ is the Kronecker product and $\psi_i$ are scalar rational functions of degree $q$. The state dimension of $\Psi$ (and matrix dimension of $P$) is $N n_b(n_u+n_y)q$ and the dimension of each $X_i$ is $(n_b+1)(n_u+n_y)$. For example $N=30$, $q=1$, $n_u=n_y=1$, $n_b = 5$ results in $P =P^\tp \in \mathbb{R}^{300 \times 300}$, each $X_i = X_i^\tp \in \mathbb{R}^{12 \times 12}$, and 47,490 decision variables, which is manageable with current SDP solvers.

\section{ADMM}
\label{sec:decomp}
The ADMM algorithm~\cite{boyd11} allows us to decompose \eqref{eq:primal} into local subproblems and a global problem involving only the supply rates and interconnection. ADMM is used to solve problems of the form
\begin{equation}\label{eq:admm_canonical_form}\begin{aligned}
\text{minimize} \qquad & f(x) + g(z) \\
\text{subject to} \qquad & Ax+Bz = c
\end{aligned}\end{equation}
where $x$ and $z$ are vector decision variables and $f$ and $g$ are extended real valued functions. The scaled ADMM algorithm is given by
\begin{align*}
x^{k+1} &= \argmin_x f(x) + \frac{\rho}{2} \|Ax+Bz^k-c+\lambda^k \|_2^2 \\
z^{k+1} &= \argmin_z g(z) + \frac{\rho}{2} \|Ax^{k+1}+Bz-c+\lambda^k \|_2^2 \\
\lambda^{k+1} &= Ax^{k+1}+Bz^{k+1}-c+\lambda^k
\end{align*}
where $\lambda$ is a scaled dual variable. The regularization parameter $\rho$ is a free parameter that typically effects the convergence rate of ADMM. However, for feasibility problems, where $f$ and $g$ are indicator functions, this parameter has no effect.

The feasibility problem~\eqref{eq:primal} may be put into this form by defining the following indicator functions:
\begin{align*}
\mathbb{I}_{\mathcal{L}_i}(X_i) &:=
	\begin{cases}
		0 &   X_i\in \mathcal{L}_i \\ 
		\infty & \text{otherwise}
	\end{cases} \\
\mathbb{I}_{\mathcal{G}}(X_1, \dots, X_N) &:=
	\begin{cases}
		0 & (X_1,\dots,X_N)\in\mathcal{G} \\ 
		\infty & \text{otherwise}
	\end{cases}
\end{align*}
and introducing the auxiliary variable $Z_i$ for each subsystem. This allows us to rewrite~\eqref{eq:primal} as an optimization problem in the canonical form~\eqref{eq:admm_canonical_form}:
\begin{align} 
\label{eq:admmForm}
\begin{split} 
\underset{(X_{1:N}, Z_{1:N})}{\text{minimize}} \quad & \sum_{i=1}^N \mathbb{I}_{\mathcal{L}_i}(X_i) +
						\mathbb{I}_{\mathcal{G}}(Z_1,\dots,Z_N) \\
\text{subject to}\quad & X_i - Z_i = 0\quad\text{for }i=1,\dots,N
\end{split}
\end{align}
where $X_{1:N} := (X_1, \dots, X_N)$. The transformation of~\eqref{eq:primal} to \eqref{eq:admmForm} is commonly used in optimization theory~\cite{boyd11} to allow decomposition of problems. Since the first term in the objective and the constraints are separable for each subsystem, the ADMM algorithm takes on the following parallelized form.
\begin{enumerate}
\item $X$-updates: for each $i$, solve the local problem:
\begin{align*}
X_i^{k+1} = \argmin_{X \in \mathcal{L}_i} 
\bigl\| X - Z_i^k + \Lambda_i^k \bigr\|_F^2
\end{align*}
\item $Z$-update: if $X_{1:N}^{k+1} \in \mathcal{G}$, then we have found a solution to \eqref{eq:primal}, so terminate the algorithm. Otherwise, solve the global problem:
\begin{align*}
Z^{k+1}_{1:N} = \argmin_{Z_{1:N}\in\mathcal{G}}  \sum_{i=1}^N \left\|X_i^{k+1} - Z_i + \Lambda_i^k\right\|_F^2
\end{align*} 
\item $\Lambda$-update: Update $\Lambda$ and return to step 1.
\[
\Lambda_i^{k+1} = X_i^{k+1} - Z_i^{k+1} + \Lambda_i^k
\]
\end{enumerate}

Figure~\ref{fig:ADMMdiagram} depicts the parallelizable nature of the ADMM algorithm. The local SDP or SOS problems are solved independently to determine the supply rate $X_i \in \mathcal{L}_i$ that is closest, in the Frobenius norm sense, to $\Lambda_i-Z_i$. These supply rates are then passed to the global SDP problem to determine the supply rates $(Z_1,\dots, Z_N) \in \mathcal{G}$ that are closest to $X_i+\Lambda_i$. The $\Lambda$-update then plays a role analogous to integral control to drive $X_i - Z_i$ to zero.

\begin{figure}[ht]
\centering
\begin{tikzpicture}[thick,auto,>=latex,node distance=2cm]

\def\s{0.6}
\def\SP{\s*2.3cm} 
\def\LWD{\s*3.2cm} 
\def\HT{\s*1.2cm} 
\def\WD{\s*10.5cm} 
\def\DX{\s*0.5cm} 

\tikzstyle{block}=[draw,rectangle,inner sep=2mm,node distance=\SP,minimum height=\HT]

\node[block,minimum width=\WD](M){$(Z_1, \dots, Z_N) \in \mathcal{G}$};
\coordinate (Mwest) at (M.south west);
\coordinate (Mmid) at (M.south);
\coordinate (Meast) at (M.south east);
\node[block,minimum width=\LWD,below of=Mwest,anchor=west](C1){$X_1 \in \mathcal{L}_1$};
\node[block,minimum width=\LWD,below of=Meast,anchor=east](CN){$X_N \in \mathcal{L}_N$};
\node[block,draw=none,below of=Mmid](Cdots){$\dots$};

\draw[->](C1.north)+(\DX,0) node (tt){} -- node[swap]{\small$X_1$} (tt |- M.south);
\draw[<-](C1.north)+(-\DX,0) node (tt){} -- node{\small$Z_1$} (tt |- M.south);

\draw[->](CN.north)+(\DX,0) node (tt){} -- node[swap]{\small$X_N$} (tt |- M.south);
\draw[<-](CN.north)+(-\DX,0) node (tt){} -- node{\small$Z_N$} (tt |- M.south);

\end{tikzpicture}
\caption{The parallelizable nature of ADMM where the local supply rates $X_i$ are updated individually based on the subsystem properties and the global supply rates $Z_i$ are updated simultaneously. \label{fig:ADMMdiagram}}
\end{figure}
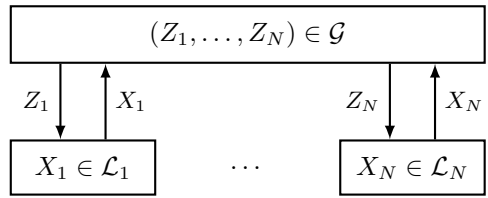

\paragraph{ADMM Convergence}The epigraph of $f$ is defined as
\if\MODE(2)
\vspace{-10 pt}
\fi
\begin{align*}
\operatorname{epi}(f) := \left\{(x,t) \in \mathbb{R}^{n+1} | f(x) \leq t \right\}.
\end{align*} 
An extended real-valued function  $f:\mathbb{R}^n \rightarrow \mathbb{R}\cup\{\infty\}$ is closed, convex, and proper if and only if the epigraph of $f$ is closed, convex, and nonempty. If the extended real-valued functions $f$ and $g$ in \eqref{eq:admm_canonical_form} are closed, proper, and convex and the Lagrangian has a saddle point then as $k \rightarrow \infty$ the objective $f(x^k) + g(z^k)$ converges to the optimal value, the dual variable $v^k$ converges to the dual optimal point $v^\star$, and the residual $Ax+Bz-c$ converges to zero~\cite{boyd11}. Furthermore, if $A$ and $B$ are full column rank then the decision variables $x^k$ and $z^k$ are guaranteed to converge to $x^\star$ and $z^\star$ as $k \rightarrow \infty$~\cite{mota11}. 

 Since the local and global constraint sets in~\eqref{eq:primal} are convex and assumed to be nonempty the indicator functions $\mathbb{I}_{\mathcal{L}_i}$ and $\mathbb{I}_{\mathcal{G}}$ are closed, convex, and proper. If the intersection of the local and global constraint sets is non-empty then a feasible point $(X^\star, Z^\star)$ exists. By Slater's condition strong duality holds. Therefore, there exists a dual optimal point $\Lambda^\star$ such that the Lagrangian has a saddle point~\cite{Bazaraa06}. Therefore, for our application ADMM is guaranteed to find a feasible point as $k \rightarrow \infty$. A feasible point is typically found in a finite number of iterations, but if the interior of the feasible set is empty the algorithm may asymptotically approach a feasible point and reach it only in the limit.
 
Other algorithms, such as alternating projections, Dykstra's method, and dual decomposition with the subgradient method, can also be used to solve \eqref{eq:primal}. A detailed comparison in~\cite{acc14} shows ADMM is significantly more reliable and typically faster than these approaches.

\section{Examples}
\label{sec:Ex}
\vspace{-8pt}
\paragraph{\bf Skew-symmetric interconnection structure} For this example $50$ LTI subsystems of the following form were generated:
\begin{align*}
G_i: \left\{ \quad 
\begin{aligned}
\dot x_i &= \bmat{-\epsilon_i & 1 \\ -1 & -\epsilon_i}x_i + \bmat{0 \\ 1}u_i \\
y_i &= \bmat{0 & 1} x_i
\end{aligned} \right.
\end{align*}
The decay rate $\epsilon_i$ was chosen from a uniform distribution over $\bmat{0,0.1}$. Each subsystem is passive, but has large $L_2$ gain due to the small decay rates $\epsilon_i$. The interconnection matrix $M$ is skew-symmetric with the following form:
\begin{align*}
M = \bmat{0 & M_0 \\ -M_0^\tp & 0}
\end{align*}
where each element of $M_0$ was chosen randomly from a standard normal distribution. This  structure arises in communication networks and multiagent systems~\cite{murat07, murat10, murat04}. Skew-symmetry of $M$ along with the passivity of the subsystems guarantees stability without any restriction on the $L_2$-gain of the subsystems or the spectral norm of $M$. In contrast, a large norm (e.g. due to the size of $M$) and the subsystem gains prevent stability certification by the small-gain theorem.

The formulation was tested on 100 random instances of the skew-symmetric interconnected system. In all instances, the ADMM algorithm certified stability in at most 65 iterations and 90\% of cases required fewer than 47 iterations (see Figure~\ref{fig:NumItersSS}).  Although we structured this example such that certification is possible with passivity, we did not bias the algorithm with this prior knowledge, and demonstrated its ability to converge to a narrow feasible set in a large-scale interconnection.  We emphasize, however, that the main interest in the algorithm is when useful structural properties of the interconnection and compatible subsystem properties are not apparent to the analyst.

\begin{figure}[ht]
	\centering
	\includegraphics{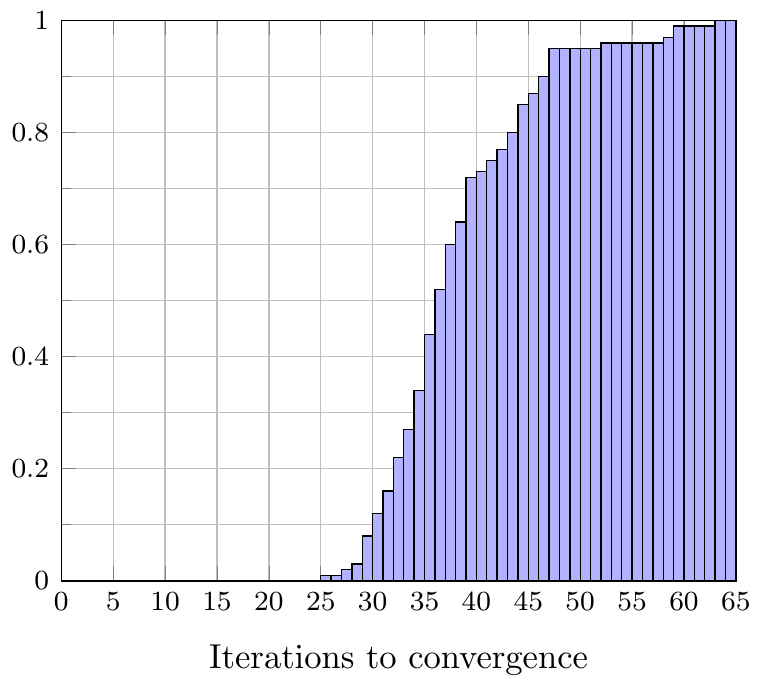}
	\caption{Cumulative plot showing the fraction of 100 total trials that required at most a given number of iterations to certify dissipativity using ADMM.
	\label{fig:NumItersSS}}
\end{figure}

\paragraph{\bf Large-scale rational polynomial system}
A system consisting of $N$, 2-state rational polynomial subsystems was generated. The subsystems are:
\begin{equation}
H:  \left\{ \hspace*{0.2in}
 \begin{aligned} \label{eq:ratdyn}
\dot{x}_{1} &= x_{2} \\
\dot{x}_{2} &= \frac{-a x_{2}-b x_{1}^3+u}{1+cx_{2}^2} \\
y &= x_{2}
\end{aligned} \right.
\end{equation}
where $a,b,c >0$ are parameters of the subsystem. The positive-definite storage function $V(x) := \frac{ab}{2} x_{1}^4 + \frac{ac}{2}x_{2}^4 + ax_2^2$ and supply rate $w(u,y) := u^2 - a^2 y^2$ certify that the $L_2$-gain of $H$ is less than or equal to $a^{-1}$.  

Clearly, {\it any} interconnection of these subsystems, with an interconnection matrix whose spectral norm is less than $a$, will have $L_2$-gain less than 1. This insight allows us to construct large-scale examples as described by the following steps:
\begin{enumerate}
\item
Choose $\left\{a_i, b_i, c_i \right\}_{i=1}^N$ uniformly distributed in $(1,2)\times(0,1)\times(0.5,2)$.
These constitute the parameters of system $H_i$.   Denote $\gamma := \max_{i} a_i^{-1}$.
\item
Choose each entry of $S \in {\mathbb{R}}^{m\times p}$ from a standard normal distribution. 
\item
Compute $\beta := \inf_B \bar{\sigma}(BMB^{-1})$ where $B = \diag(b_1, \dots, b_N, I_d)$, $b_i > 0$ for $i =1,\dots, N$. Redefine $S:= \frac{0.99}{\gamma \beta}S$ to guarantee the spectral norm of the interconnection is less than $\gamma$.
\item
Choose random nonzero, diagonal scalings $\Psi = \diag(\Psi_1,\dots,\Psi_N)$ 
and $\Phi = \diag(\Phi_1,\dots,\Phi_N)$. 
\item
Define $G_i := \Phi_i H_{i} \Psi_i$, and
\[
M :=
\left[ \begin{array}{cc} \Psi^{-1} & 0 \\ 0 & I_d \end{array} \right]
S 
\left[ \begin{array}{cc} \Phi^{-1} & 0 \\ 0 & I_d \end{array} \right].
\]
\end{enumerate}
The scalings introduced in step 4 alter the gain properties of the subsystems and interconnection disguising the simple construction that guarantees the $L_2$-gain of the interconnected system is less than 1. Figure ~\ref{fig:Scalings} below illustrates the interconnection that the algorithm must attempt to certify.

 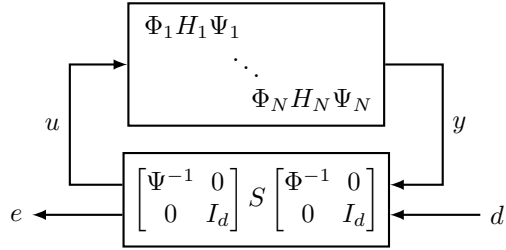
\begin{figure}[ht]
 \centering
 \begin{tikzpicture}[thick,auto,>=latex,node distance=1.8cm]
 \tikzstyle{block}=[draw,rectangle,minimum width=34mm, minimum height=2.5em]
 \def\x{0.7}		
 \def\y{0.2}		
 \def\z{1.2}		
 \node [block,inner ysep=2mm](G){$\addtolength{\arraycolsep}{-0.3em}\bmat{\Psi^{-1}&0 \\0 & I_d}S\bmat{\Phi^{-1}& 0\\0 & I_d}$};
 \node [block,above of=G](D){$\addtolength{\arraycolsep}{-0.7em}\begin{matrix}
 \,\Phi_1{H}_1\Psi_1 & &\\[-0.15em]&\ddots&\\[-0.15em]& & \Phi_N{H}_N\Psi_N \end{matrix}$};
 \draw [<-] (G.east)+(0,\y) -- +(\x,\y) |- node[swap,pos=0.25]{$y$} (D);
 \draw [<-] (G.east)+(0,-\y) -- +(\z,-\y) node [anchor=west]{$d$};
 \draw [->] (G.west)+(0,\y) -- +(-\x,\y) |- node[pos=0.25]{$u$} (D);
 \draw [->] (G.west)+(0,-\y) -- +(-\z,-\y) node [anchor=east]{$e$};
 \end{tikzpicture}
 \caption{Scaling of the interconnected system of Figure~\ref{fig:IntSysIO}.\label{fig:Scalings}}
 \end{figure}

We generated 200 random instances of this interconnected system, each with $N=100$ subsystems. The ADMM algorithm was applied to certify the $L_2$-gain of the interconnection is less than or equal to 1. SOS programming was used to search for quartic storage functions to certify dissipativity of the subsystems. Each storage function consists of all monomials up to degree 4. The algorithm succeeded for all 200 tests, requiring at most 48 iterations and less than 15 for 90\% of the tests.

Using this example we also were able to test the performance of our method compared to directly searching for a separable storage function. Both polynomial and rational polynomial systems with different numbers of subsystems were tested. The subsystems were generated as described above. For the polynomial subsystems the coefficient $c$ in~\eqref{eq:ratdyn} was set to~$0$. Since the number of iterations for the ADMM algorithm may vary, 100 tests for each system size were performed. Figure~\ref{fig:admmTimeComp} shows the average time for the ADMM algorithm to find a solution compared to the time to directly search for a separable storage function.

 \begin{figure}[h!]
    \centering
      \includegraphics{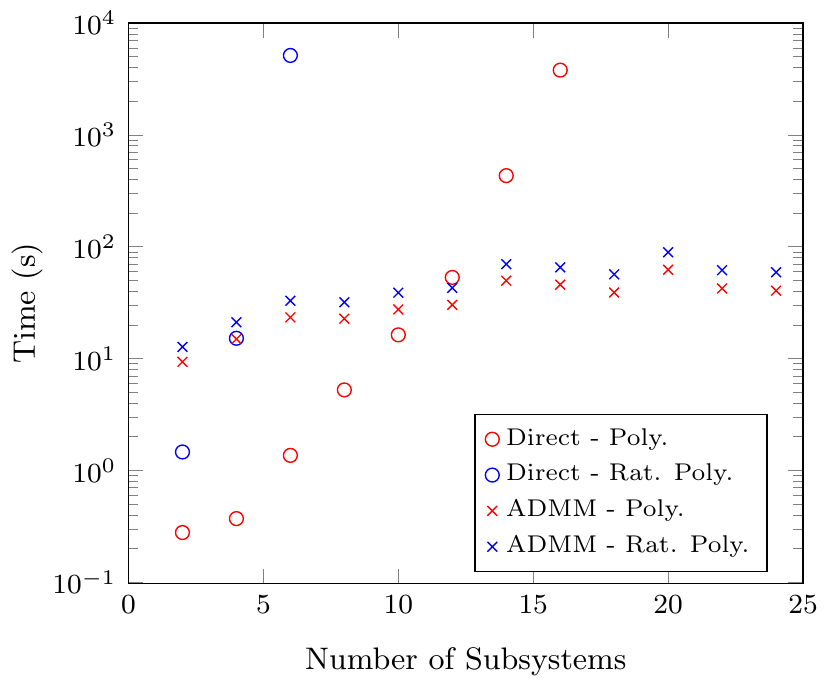}
      \caption{Runtime of the proposed method using ADMM compared to directly finding a separable storage function.} \label{fig:admmTimeComp}
  \end{figure}
 
 As can be seen for large systems the ADMM algorithm outperforms conventional approaches. Directly searching for a separable storage function became computationally intractable for systems with more than 16 polynomial subsystems and more than 6 rational subsystems, while the ADMM algorithm has been used to certify properties for systems with 200 rational subsystems.

\paragraph{\bf Vehicle platoon}
In this example we analyze the $L_2$-gain properties of a vehicle platoon~\cite{burger13,coogan14}. For a platoon with $N$ vehicles the dynamics of the $i^\text{th}$ vehicle can be described by
\begin{equation*}
\Sigma_i: \left\{
\begin{aligned}
\quad \dot{v}_i &= -v_i + v_i^{\textup{nom}} + u_i \\
y_i &= v_i 
\end{aligned} \right.
\qquad
i = 1,\dots,N
\end{equation*}
where $v_i(t)$ is the vehicle velocity and $v_i^{\textup{nom}}$ is the nominal velocity. In the absence of a control input $u_i(t)$ each vehicle tends to its nominal velocity.

\begin{figure}[ht]
\centering
\begin{tikzpicture}[thick,auto,>=latex,node distance=2cm]
\tikzstyle{dpath}=[dashed,very thick,draw=red]
\def\s{0.6}
\def\car(#1)#2#3{
  \begin{scope}[shift={(#1)}]
    \def\w{\s*1cm};  \def\h{\s*0.5cm}; \def\tx{\s*0.15cm}
    \def\th{\s*0.12cm}; \def\tw{\s*0.5cm}
    \draw (-\w,-\h) rectangle (\w,\h);
    \draw (-\w+\tx,-\h) rectangle (-\w+\tx+\tw,-\h-\th);
    \draw (-\w+\tx,\h) rectangle (-\w+\tx+\tw,\h+\th);
    \draw (\w-\tx,-\h) rectangle (\w-\tx-\tw,-\h-\th);
    \draw (\w-\tx,\h) rectangle (\w-\tx-\tw,\h+\th);
    \draw (0,-1.1*\h-0.1) -- (0,-2.5*\h);
    \draw[->,swap,pos=0.5] (0,-2.0*\h) -- node {$x_{#3}$}  (1,-2.0*\h);
    \node[minimum width=2*\w, minimum height=\w+2*\th] (#2) at (0,0) {$#3$};
  \end{scope}
}
\car(0,0){A}{3}
\car(3,0){B}{2}
\car(6,0){C}{1}
\path [dpath] (A) -- (B) -- (C);
\end{tikzpicture}
\caption{Vehicle platoon. Each vehicle measures the relative distance of all vehicles connected to it by a dotted line.\vspace{1mm} \label{fig:PlatoonEx}}
\end{figure}
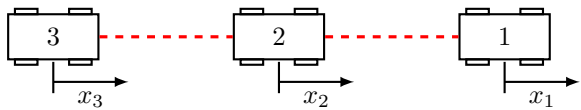

Each vehicle uses the relative distance between itself and a subset of the other vehicles to control its velocity. The subsets are represented by a connected, bidirectional, acyclic graph with $L$ links interconnecting the $N$ vehicles. In Figure~\ref{fig:PlatoonEx}, the links are shown as dotted lines. Letting $p_\ell$ be the relative displacement between the vehicles connected by link $\ell$ gives $\dot{p}_{\ell} = v_i-v_j$ where $i$ is the leading node and $j$ is the trailing node. We define $D \in \mathbb{R}^{N \times L}$ as 
\begin{equation*}
D_{i\ell} = \left\{ \begin{array}{ll} 1 & \text{if $i$ is the leading node of edge $\ell$}\\
-1 & \text{if $i$ is the trailing node of edge $\ell$}\\ 
0 &\text{otherwise.} \end{array} \right.
\end{equation*}
Thus, $D$ maps the velocities of the vehicles to the relative velocities across each link: $\dot{p} = D^\tp v$.

 We consider a set of control laws for velocity regulation that encompass those presented in~\cite{burger13,coogan14}:
\vspace{-12 pt}
\begin{align*}
u_i = -\sum_{\ell=1}^LD_{i\ell}\phi_\ell(p_{\ell})
\end{align*}
where $\phi_\ell:\mathbb{R}\to\mathbb{R}$ can be any function that is increasing and surjective, ensuring the existence of an equilibrium point~\cite{burger13}. Defining $\Phi := \diag(\phi_1, \dots, \phi_L)$, we represent the system as the block diagram in Figure~\ref{fig:PlatoonBD}.

 \begin{figure}[ht]
 \centering
 \begin{tikzpicture}[thick,auto,>=latex,node distance=2cm]
 \tikzstyle{block}=[draw,rectangle,minimum width=8mm, minimum height=2em]
 \def\x{0.7}		
 \def\y{0}		
 \def\z{1.2}		
  \node [block,inner xsep=1.5mm, inner ysep=1.5mm](G){$\addtolength{\arraycolsep}{-0.3em}\begin{matrix}
  \Sigma_1 & &\\[-1.5mm] & \ddots & \\[-1mm] & & \Sigma_N \end{matrix}$};
 \node [draw=none,fill=none, below of=G,node distance=8mm](M){};
 \node [block, left of=M, node distance=20mm](D){$ -D$};
 \node [block, right of=M, node distance=20mm](DT){$ D^\tp$};
 \node [draw=none,fill=none, below of=G](N){};
 \node [block, right of=N, node distance=8mm](S){$\int$};
 \node [block, left of=N, node distance=8mm](P){$\Phi$};
 
 \draw [->] (G.east)+(0,\y) -- +(0,\y) -| node[pos=0.25]{$v$} (DT);
 \draw [->] (D.north)+(0,\y) -- +(0,\y) |- node[pos=0.75]{$u$} (G);
 \draw [->] (DT.south)+(0,-\y) -- +(0,-\y)|- node[pos=0.3]{$\eta$} node[pos=0.8]{$\dot{p}$}  (S);
 \draw [->] (S) -- node[pos=0.5]{$p$} (P);
 \draw [->] (P.west)+(0,\y) -- +(0,\y) -| node[pos=0.7]{$z$}  (D.south);
 \path (P.north west)+(-.3, .25) coordinate (R1);
 \path (S.south east)+(.6,-.25) coordinate (R2);
 \node [draw = none, fill=none, left of=R2, node distance=38mm] (lam){};
 \draw[dashed] (R1) rectangle (R2);
 \node[anchor=south west, xshift = 3.8cm, yshift = -0.2cm] at (lam) {$\Lambda$};
  \end{tikzpicture}
 \caption{Block diagram of the vehicle platoon dynamics. \label{fig:PlatoonBD}}
 \end{figure}
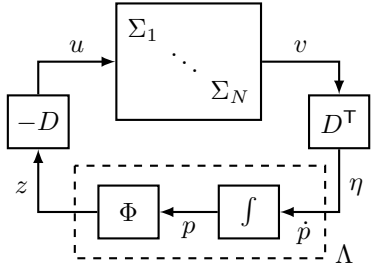
 
 The map $\Lambda$ from $\eta$ to $z$, indicated by a dashed box in Figure~\ref{fig:PlatoonBD}, is diagonal; each $\dot p_\ell$ is integrated and the corresponding $\phi_\ell$ is applied. Thus, we define $\Lambda := \diag(\Lambda_1, \dots, \Lambda_L)$, where $\Lambda_\ell$ is
 \begin{equation*}
 \Lambda_\ell: \left\{
 \begin{aligned}
 \quad \dot{p}_\ell &= \eta_\ell \\
 z_\ell &= \phi_\ell(p_\ell)
 \end{aligned}
 \qquad \ell = 1,\dots,L \right.
 \end{equation*}
 with input $\eta_\ell$ and output $z_\ell$. 
 By diagonally concatenating  $\Sigma := \diag(\Sigma_1, \dots, \Sigma_N)$ with $\Lambda$ we can transform this system into Figure~\ref{fig:PlatoonTrans}.  

 \begin{figure}[ht]
 \centering
 \begin{tikzpicture}[thick,auto,>=latex,node distance=1.8cm]
 \tikzstyle{block}=[draw,rectangle,minimum width=20mm, minimum height=2.5em]
 \def\x{0.7}		
 \def\y{0}		
 \def\z{1.2}		
 \node [block,inner ysep=2mm](G){$\addtolength{\arraycolsep}{-0.3em}\bmat{0 & -D \\ D^\tp & 0}$};
 \node [block,above of=G](D){$\addtolength{\arraycolsep}{-0.3em}\bmat{
 \Sigma & 0 \\ 0 & \Lambda}$};
 
 \draw [<-] (G.east)+(0,\y) -- +(\x,\y) |- node[swap,pos=0.25]{$\bmat{v \\ \phi(p)}$} (D);
 \draw [->] (G.west)+(0,\y) -- +(-\x,\y) |- node[pos=0.25]{$\bmat{u \\ \dot{p}}$} (D);
 \end{tikzpicture}
 \caption{Figure~\ref{fig:PlatoonBD} transformed into the form of Figure~\ref{fig:IntSysIO}. \label{fig:PlatoonTrans}}
 \end{figure}
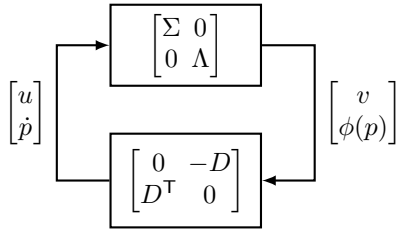

An equilibrium $(v^\star, p^\star)$ is guaranteed to exist, but it depends on the unknown functions $\phi_\ell$. Therefore, we will exploit the EID properties of the subsystems to establish the desired global property without explicit knowledge of the equilibrium. 

For each $\Lambda_\ell$ subsystem the dissipativity properties depend on the unknown $\phi_\ell$ function. However, it is not difficult to show that $\Lambda_i$ is EID with respect to the following supply rate 
\begin{equation}\label{eq:passivity}
\bmat{\eta_\ell-\eta_\ell^\star \\ z_\ell-z_\ell^\star}^\tp \bmat{0 & 1 \\ 1  & 0} \bmat{\eta_\ell-\eta_\ell^\star \\ z_\ell-z_\ell^\star}.
\end{equation} 
This property can be proven by using the storage function
\begin{equation*}
V_\ell(p_\ell) = 2\int_{p_\ell^\star}^{p_\ell} [\phi_{\ell}(\theta) - \phi_\ell(p_\ell^\star)]\,d\theta
\end{equation*}
and the property $(p_\ell-p_\ell^\star) [\phi_\ell(p_\ell)-\phi_\ell(p_\ell^\star)] \geq 0$ which follows because $\phi_\ell$ is increasing. Therefore, instead of searching over supply rates for the $\Lambda_\ell$ subsystems in the ADMM algorithm, we fix~\eqref{eq:passivity} as the supply rate so that the algorithm does not rely on the $\phi_\ell$ functions or their associated equilibrium points.

For the simulation, we used $N=20$ and each vehicle's nominal velocity was randomly chosen. A linear topology was used as in Figure~\ref{fig:PlatoonEx}. That is, each vehicle measures the distance to the vehicle in front of it and the vehicle behind it. We investigated how a force disturbance applied to the trailing vehicle would affect the velocity of the lead vehicle. Specifically, we augmented the interconnection matrix $M$ (see Figure~\ref{fig:IntSysIO}) such that the disturbance $d$ is applied to the last vehicle:
\begin{align*}
\dot{v}_N = -v_N + v_N^{\textup{nom}} + u_N+d
\end{align*} 
and the output $e$ is the velocity of the first vehicle $v_1$. We then certified the $L_2$-gain from $d$ to $e$ is no greater than $\gamma$ using the supply rate 
\begin{align*}
\bmat{d \\ e-e^\star}^\tp \bmat{ 1 & 0 \\ 0 & -\gamma^{-2} } \bmat{d \\ e-e^\star}.
\end{align*}
 A bisection search was used to find that $\gamma_\textup{min} = 0.71$ was the smallest value that could be certified. Since our method searches over a restricted class of possible storage functions it may be conservative. To bound this conservatism, we performed an ad-hoc search over linear $\phi_\ell$ functions, seeking a worst-case $L_2$-gain. The result was that $\gamma_\textup{min} \geq 0.49$.

In this problem the interior of the feasible set is empty. Letting $X^{jk}$ be the entry in the $j$-th row and $k$-th column of $X$, the local subproblems have the constraints $X^{11}_i \geq 0$ for $i=1,\dots,N$ (these are scalar variables), while the global problem has the constraint $D^\tp \!\diag(Z^{11}_1,\dots,Z^{11}_N) D \leq 0$. The only solution that satisfies both of these constraints is $X^{11}_i=Z^{11}_i=0$ for all $i$. This further implies that $X^{12}_i=X^{21}_i=1$ for all $i$. The intersection of the $\mathcal{G}$ and $\mathcal{L}_i$ sets having an empty interior results in the ADMM algorithm oscillating between $X^{11}_{i} > 0$ for the local problems and $Z^{11}_i < 0$ for the global problem, leading to slow convergence and finding a feasible solution (i.e. $X^{11}_{i} =Z^{11}_{i}= 0$) in the limit. 

The difficulties arising from hidden equality constraints in SDP problems are well known and there are procedures for automatically detecting these~\cite{parillo14}. Unfortunately, it is not clear how to apply these ideas here, because the equality constraint is only present when the local and global constraints are considered simultaneously. We addressed this issue by setting $X_{i}^{11}=0$ and $X_i^{12}=X_i^{21}=1$, effectively removing those variables from the optimization. The feasible region of the resulting problem has a nonempty interior, and the ADMM algorithm converges in a few iterations.

\paragraph{\bf IQC example}

This example, while very simple, demonstrates the advantage of using IQCs instead of dissipativity. Consider two subsystems $G_1$ and $G_2$ interconnected by $M$:
\begin{align*}
G_1(s) = \frac{1}{s+1} \quad
G_2(s) = \frac{2}{5s+1} \quad M = \bmatshort{0 & -1 & 1 \\ 1 & 0 & 0 \\ 1 & 0 & 0}.
\end{align*}

The $L_2$-gain of the interconnected (linear) system is approximately $0.862$. 
However, using duality certificates, one can show that the dissipativity formulation {\it cannot} certify that the $L_2$-gain of the interconnected system is less than 1.   By contrast, with the IQC formulation, and
\begin{align*}
\Psi(s) = \bmat{\psi_0(s) \\ \vdots \\ \psi_K(s)}\otimes I_2 \text{ where } \psi_k(s) = \left(\frac{1}{s+2}\right)^k
\end{align*}
and $K=3$, the ADMM algorithm certifies the $L_2$-gain of the interconnected system is less than or equal to $0.863$. 

\par All simulations were performed in MATLAB using the SOSOPT toolbox \cite{sosopt} for SOS programs and the CVX toolbox \cite{cvx} for SDP problems.

\section{Conclusion}
A compositional approach to performance certification of large interconnected systems was presented. This approach is less conservative than conventional techniques because it searches for the subsystem properties that are most beneficial in certifying the performance of the interconnection. For the linear case, we have shown this approach is equivalent to searching for a separable storage function. ADMM is used to decompose this problem enabling certification of much larger problems than conventional techniques allow.

\begin{small}
\bibliographystyle{abbrv}
\bibliography{autobib}
\end{small}

\newpage
\section*{Appendix: Proof of Theorem \ref{thm:equivLin}}
\label{sec:app}

(\ref{cond2})$\implies$(\ref{cond1}) follows by specializing Proposition~\ref{prop:generalized_inequality} to linear subsystems and quadratic storage functions of the form $V_i(x_i) = x_i^\tp P_i x_i$. Then, the dissipation inequality~\eqref{eq:diss} implies condition (\ref{cond1}).

(\ref{cond1})$\implies$(\ref{cond2}): Condition (\ref{cond1}) is equivalent to the existence of $P_i \succeq 0$ for $i=1,\dots,N$ such that
\begin{equation} \label{eq:LTIdiss}
\hspace{-3mm}\sum_{i=1}^N \bmat{x_i\\u_i}^\tp\!\! \addtolength{\arraycolsep}{-3pt}\bmat{A_i^\tp P_i+P_iA_i & P_iB_i \\ B_i^\tp P_i & 0 }\!\bmat{x_i \\ u_i} \leq \bmat{d \\ e}^\tp\!\! W \bmat{d \\ e}
\end{equation}
for all $x$ and $d$, where $u$ and $e$ are expressed in terms of $x$ and $d$ using~\eqref{eq:M_interconnection} and~\eqref{subsysb}.

Defining $V_i(x_i) := x_i^\tp P_i x_i$, the $i^\text{th}$ summand on the left-hand side of~\eqref{eq:LTIdiss} is  $\dot V(x_i, u_i):=\nabla V_i(x_i)^\tp f_i(x_i, u_i)$. We will prove the $V_i$ are storage functions that certify local dissipativity of the subsystems. We assume without loss of generality that for each subsystem, $C_i= \bmatshort{\mathbf{0}_{m_i \times (n_i-m_i)} & \mathbf{I}_{m_i \times m_i}}$. This allows each $x_i$ to be partitioned as $x_i=\bmatshort{z_i \\ y_i}$ where $z_i \in \mathbb{R}^{n_i-m_i}$ and $y_i \in \mathbb{R}^{m_i}$. Eliminating $x_i$ from \eqref{eq:LTIdiss} and rearranging, we obtain
\if\MODE1
\begin{align}\label{eq:DIElin}
\sum_{i=1}^N \left(z_i^\tp  Q_iz_i+2z_i^\tp R_i \bmat{ u_i\\ y_i} +\bmat{ u_i \\ y_i}^\tp S_i \bmat{ u_i\\ y_i} \right) \leq \bmat{d \\ e}^\tp W \bmat{d \\ e} \quad \text{ for all $z,y,d$}
\end{align}
\else
\begin{align}\label{eq:DIElin}
\sum_{i=1}^N &\left(z_i^\tp  Q_iz_i+2z_i^\tp R_i \bmat{ u_i\\ y_i} +\bmat{ u_i \\ y_i}^\tp S_i \bmat{ u_i\\ y_i} \right) \notag \\ & \leq \bmat{d \\ e}^\tp W \bmat{d \\ e} \quad \text{ for all $z,y,d$}
\end{align}
\fi
and appropriately chosen $Q_i$, $R_i$, and $S_i$. Since~\eqref{eq:DIElin} holds for all $z,y,d$, it holds in particular when $y=d=0$. We then have from~\eqref{eq:M_interconnection} that $u=e=0$, and we conclude that $Q_i \preceq 0$. Using a similarity transform, we may assume, again without loss of generality, that $z_i$ can be decomposed as:
\begin{align*}
z_i = \bmat{\xi_i \\ \hat{z}_i}, \text{ where } z_i^\tp Q_i z_i = \hat{z}_i^\tp \hat{Q}_i \hat{z}_i \text{ and } \hat{Q}_i \prec 0.
\end{align*}
The dimensions of $\hat{z}_i$ and $\hat{Q}_i$ correspond to the number of nonzero eigenvalues of $Q_i$. Rewriting
\begin{align*}
z_i^\tp R_i\bmat{u_i \\ y_i} = \xi_i^\tp Y_i y_i +\xi_i^\tp U_i u_i+ \hat{z}_i^\tp \hat{R}_i\bmat{u_i \\ y_i}
\end{align*}
where $Y_i$, $U_i$, and $\hat{R}_i$ are appropriately defined matrices, the summands in~\eqref{eq:DIElin} take the form
\if\MODE1
\begin{equation}\label{etmp}
\dot{V}_i(\hat{z}_i, \xi_i, y_i, d) = \hat{z}_i^\tp \hat{Q}_i \hat{z}_i + 2\xi_i^\tp Y_i y_i +2\xi_i^\tp U_i u_i
+ 2\hat{z}_i^\tp \hat{R}_i\bmat{u_i \\ y_i}+ \bmat{u_i \\ y_i}^\tp S_i \bmat{u_i \\ y_i}.
\end{equation}
\else
\begin{multline}\label{etmp}
\dot{V}_i(\hat{z}_i, \xi_i, y_i, d) = \hat{z}_i^\tp \hat{Q}_i \hat{z}_i + 2\xi_i^\tp Y_i y_i +2\xi_i^\tp U_i u_i \\
+ 2\hat{z}_i^\tp \hat{R}_i\bmat{u_i \\ y_i}+ \bmat{u_i \\ y_i}^\tp S_i \bmat{u_i \\ y_i}.
\end{multline}
\fi
Because~\eqref{eq:DIElin} holds for all $z,y,d$ it must also hold if we maximize over $\hat z_i$. Performing the maximization,
\if\MODE1
\begin{equation} \label{eq:maxZ} 
\dot{V}_i(\hat{z}_i^\star, \xi_i, y_i, d) = \bmat{ u_i \\ y_i}^\tp (S_i-\hat{R}_i^\tp \hat{Q}_i^{-1} \hat{R}_i) \bmat{ u_i\\ y_i} 
+2\xi_i^\tp Y_i y_i +2\xi_i^\tp U_i u_i
\end{equation}
\else
\begin{multline} \label{eq:maxZ} 
\dot{V}_i(\hat{z}_i^\star, \xi_i, y_i, d) = \bmat{ u_i \\ y_i}^\tp (S_i-\hat{R}_i^\tp \hat{Q}_i^{-1} \hat{R}_i) \bmat{ u_i\\ y_i} \\ 
+2\xi_i^\tp Y_i y_i +2\xi_i^\tp U_i u_i
\end{multline}
\fi
where $\hat z^\star_i := \arg\max_{\hat z_i} \dot V_i(\hat z_i,\xi_i,y_i,d)$.
If we further define $X_i := S_i-\hat{R}_i^\tp \hat{Q}_i^{-1} \hat{R}_i$, we can write
\if\MODE1
\begin{equation}\label{eq:dissbar}
\sum_{i=1}^N \dot{V}_i(\hat{z}^\star_i, \xi_i, y_i, d) = \sum_{i=1}^N\bmat{u_i \\ y_i}^\tp X_i\bmat{u_i \\ y_i}
+ 2\xi^\tp\addtolength{\arraycolsep}{-2pt}\bar Y y+
2\xi^\tp\bar Uu
\end{equation}
\else
\begin{multline}\label{eq:dissbar}
\sum_{i=1}^N \dot{V}_i(\hat{z}^\star_i, \xi_i, y_i, d) = \sum_{i=1}^N\bmat{u_i \\ y_i}^\tp X_i\bmat{u_i \\ y_i} \\
+ 2\xi^\tp\addtolength{\arraycolsep}{-2pt}\bar Y y+
2\xi^\tp\bar Uu
\end{multline}
\fi
where $\bar Y := \diag(Y_1,\dots,Y_N)$ and $\bar U$ is similarly defined. Thus,
\if\MODE1
\begin{align}\label{eq:ded}
\sum_{i=1}^N \dot{V}_i(\hat{z}_i, \xi_i, y_i,d)
\le \sum_{i=1}^N \dot{V}_i(\hat{z}_i^\star, \xi_i, y_i,d)
\le \bmat{d \\ e}^\tp W \bmat{d \\ e}
\end{align}
\else
\begin{align}\label{eq:ded}
\begin{split}
\sum_{i=1}^N \dot{V}_i(\hat{z}_i, \xi_i, y_i,d)
&\le \sum_{i=1}^N \dot{V}_i(\hat{z}_i^\star, \xi_i, y_i,d) \\
&\le \bmat{d \\ e}^\tp W \bmat{d \\ e}
\end{split}
\end{align}
\fi
for all $\hat z,\xi,y,d$ where $e,u$ satisfy~\eqref{eq:M_interconnection}. Note that the the right-hand side of~\eqref{eq:ded} does not depend on $\xi$, yet its lower bound~\eqref{eq:dissbar} is linear in $\xi$. The only way this inequality can be true for all~$\xi$ is if
\begin{equation*}
2\xi^\tp\bar Yy+ 2\xi^\tp\bar Uu = 0 
\qquad\text{for all $\xi$, $y$, $d$}.
\end{equation*}
From~\eqref{eq:M_interconnection}, we have $u = M^{11}y + M^{12}d$ and so
\begin{align} \label{eq:YUconst}
\bar Y+ \bar U M^{11} = \bar U M^{12} = 0.
\end{align}

By denoting $M^{11}_{ij}$ as the submatrix of $M$ mapping $y_j \rightarrow u_i$ and $M^{12}_{j}$ as the submatrix of $M$ mapping $d \rightarrow u_i$, then for each $i$, \eqref{eq:YUconst} simplifies to
\begin{align*}
Y_i +U_iM^{11}_{ii} &= 0 \qquad \\
U_i M^{11}_{ij} &= 0 \quad \text{for $j \neq i$} \\
U_i M_j^{12} &=0 \quad \text{for $j \in 1,\dots, N$}
\end{align*} 
Assumption~\ref{assum:Mind} implies that $U_i=Y_i=0$. Therefore, \eqref{eq:maxZ} simplifies to
\begin{equation}\label{e:fin1}
\dot V_i(x_i, d) \le  \dot V_i(\hat z_i^\star,\xi_i,y_i, d) = \bmat{u_i \\ y_i }^\tp X_i \bmat{u_i \\ y_i}
\end{equation}
and hence, the storage function $V_i$ certifies dissipativity of the $i^\text{th}$ local subsystem with respect to the supply rate matrix $X_i$. Combining~\eqref{eq:ded} and~\eqref{e:fin1},
\begin{equation}\label{e:fin2}
\sum_{i=1}^N \bmat{ u_i\\ y_i}^\tp X_i \bmat{ u_i\\ y_i} \leq \bmat{d \\ e}^\tp W \bmat{d \\ e}.
\end{equation}
It follows from~\eqref{e:fin1}--\eqref{e:fin2} that each subsystem is dissipative and  $X_1, \dots, X_N$ satisfies~\eqref{def:G}. \hfill \qed
\end{document}